\documentclass[conference,10pt]{IEEEtran}
\IEEEoverridecommandlockouts
\usepackage{cite}
\usepackage{subfigure}
\usepackage{amsmath,amssymb,amsfonts,stmaryrd}
\usepackage{amsthm,mathtools}
\usepackage{algorithmic}
\usepackage{algorithm}
\usepackage{graphicx}
\usepackage{textcomp}
\usepackage{xcolor}
\usepackage[normalem]{ulem}
\usepackage{array}
\usepackage{xparse} \DeclarePairedDelimiterX{\Iintv}[1]{\llbracket}{\rrbracket}{\iintvargs{#1}}
\NewDocumentCommand{\iintvargs}{>{\SplitArgument{1}{,}}m}
{\iintvargsaux#1} %
\NewDocumentCommand{\iintvargsaux}{mm} {#1\mkern1.5mu..\mkern1.5mu#2}
\newtheorem{theorem}{Theorem}
\newtheorem{lemma}{Lemma}

\newtheorem{assumption}{Assumption}
\newtheorem{corollary}{Corollary}
\newtheorem{definition}{Definition}
\def\BibTeX{{\rm B\kern-.05em{\sc i\kern-.025em b}\kern-.08em
    T\kern-.1667em\lower.7ex\hbox{E}\kern-.125emX}}
    
\begin{document}

\title{Affine Frequency Division Multiplexing for Compressed Sensing of Time-Varying Channels}

\author{
Wissal Benzine\textit{$^{1,2}$}, Ali Bemani\textit{$^1$}, Nassar Ksairi\textit{$^1$}, and Dirk Slock\textit{$^2$} \\
\textit{$^1$}Mathematical and Algorithmic Sciences Lab, Huawei France R\&D, Paris, France \\ \textit{$^2$}Communication Systems Department, EURECOM, Sophia Antipolis, France\\
Emails:
\{wissal.benzine1, ali.bemani, nassar.ksairi\}@huawei.com, Dirk.Slock@eurecom.fr
}


\maketitle

\begin{abstract}
 This paper addresses compressed sensing of linear time-varying (LTV) wireless propagation links under the assumption of \emph{double sparsity} i.e., sparsity in both the delay and Doppler domains, using Affine Frequency Division Multiplexing (AFDM) measurements. By rigorously linking the \emph{double sparsity} model to the \emph{hierarchical sparsity} paradigm, a compressed sensing algorithm with recovery guarantees is proposed for extracting delay-Doppler profiles of LTV channels using AFDM. Through mathematical analysis and numerical results, the superiority of AFDM over other waveforms in terms of channel estimation overhead and minimal sampling rate requirements in sub-Nyquist radar applications is demonstrated.

\end{abstract}


\section{Introduction}
Time-varying wireless channels in many propagation scenarios, especially in high-frequency bands, are characterized by sparsity in both delay and Doppler domains \cite{sparse_sbl}. Such sparsity is an important feature of wireless propagation that can be exploited to improve channel estimation performance \cite{sparse_spawc} or radar sensing \cite{Eldar_sub_Nyquist}. Delay-Doppler sparsity was assumed in \cite{sparse_sbl,kron_sbl} and leveraged to conceive enhanced channel estimation schemes for time-varying channels using the sparse Bayesian learning (SBL) framework. However, delay-Doppler sparsity was modeled as the sparsity of a one-dimensional array with no way to assign different sparsity levels to the delay and Doppler domains. To obtain a sparsity model compatible with the latter requirement, one can in principle turn to the \emph{hierarchical sparsity} framework \cite{hierarchical}. Indeed, the concept and the tools of hierarchical sparsity were applied in \cite{hierarchical_ce} to the problem of multi-input multi-output (MIMO) channel estimation under delay and angular domains sparsity.
In sensing and radar applications, the sub-Nyquist radar paradigm \cite{Eldar_sub_Nyquist} leverages wireless channel sparsity to develop sub-Nyquist receivers. However, most of its solutions cannot take advantage of Doppler domain sparsity for lowering the sampling rate and some of them require complex analog-domain processing.

In \cite{AFDM_letter}, the relevance of affine frequency division multiplexing (AFDM) \cite{BemaniAFDM_TWC}, a recently proposed waveform based on the discrete affine Fourier transform (DAFT), for efficient self-interference cancellation in mono-static integrated sensing and communications (ISAC) scenarios was demonstrated. In \cite{afdm_gc}, we had established its relevance for time-varying channel estimation under delay and Doppler sparsity with a known delay-Doppler profile (DDP). Using tools from the framework of hierarchical sparsity, the current work tackles the problem of delay-Doppler sparse recovery when no such DDP knowledge is assumed, with applications to both time-varying channel estimation and sub-Nyquist radar sensing.

\subsection{Contributions}
I) The statistical notion of delay-Doppler sparsity is rigorously linked to the hierarchical sparsity paradigm. II) This link is used to propose a sparse recovery algorithm based on AFDM measurements for delay-Doppler profiles of wireless channels. III) Using hierarchical-sparsity mathematical tools, closed-form asymptotic results for the performance of this recovery is provided. Finally, IV) this performance analysis is used to show the superiority of AFDM over recovery schemes based on other waveforms in terms of channel estimation overhead and sensing receiver minimal sampling rate requirements.
\subsection{Notations}
$\mathrm{Bernoulli}(p)$ is the Bernoulli distribution with probability $p$ and $\mathrm{B}(n,p)$ is the binomial distribution with parameters $(n,p)$.
Notation $X\sim F$ means that random variable $X$ has distribution $F$.
If $\mathcal{A}$ is a set, $|\mathcal{A}|$ stands for its cardinality. The set of all integers between $l$ and $m$ (including $l$ and $m$, $(l,m)\in\mathbb{Z}^2$) is denoted $\Iintv{l,m}$.
The ceiling operation is denoted as $\lceil.\rceil$. The modulo $N$ operation is denoted as $(\cdot)_N$.

\section{Background: AFDM}
In AFDM, modulation is achieved through the use of DAFT which is a discretized version \cite{erseghe2005multicarrier} of the affine Fourier transform (AFT) \cite{BemaniAFDM_TWC} with the discrete chirp $e^{-\imath2\pi (c_2k^2+{\frac{1}{N} }kn+c_1n^2)}$ as its kernel (see Fig. \ref{fig:afdm_t_f}). Here, $c_1$ and $c_2$ are parameters that can be adjusted depending on the delay-Doppler profile of the channel (in this work, we adjust $c_1$ based on the delay-Doppler sparsity levels).
Consider a set of quadrature amplitude modulation (QAM) symbols $\{x_k\}_{k=0\cdots N-1}$. AFDM employs inverse DAFT (IDAFT) to map $\{x_k\}_{k=0\cdots N-1}$ to 
\begin{equation}
\label{eq:AFDM_mod}
    s_n = \frac{1}{\sqrt{N}}\sum_{k = 0}^{N-1}x_ke^{\imath2\pi (c_2k^2+{\frac{1}{N}}kn+c_1n^2)}, n= 0\cdots N-1
\end{equation}
with the following so called {\emph{chirp-periodic prefix}} (CPP)
\begin{equation}
    s_n = s_{N+n}e^{-\imath2\pi c_1(N^2+2Nn)},\quad n = -L_{\mathrm{CPP}}\cdots -1
\end{equation}
where $L_{\mathrm{CPP}}$ denotes an integer that is greater than or equal to the number of samples required to represent the maximum delay of the wireless channel. The CPP simplifies to a cyclic prefix (CP) whenever $2c_1N$ is integer and $N$ is even, an assumption that will be considered to hold from now on.
\begin{figure}
    \centering
    \includegraphics[width=\columnwidth]{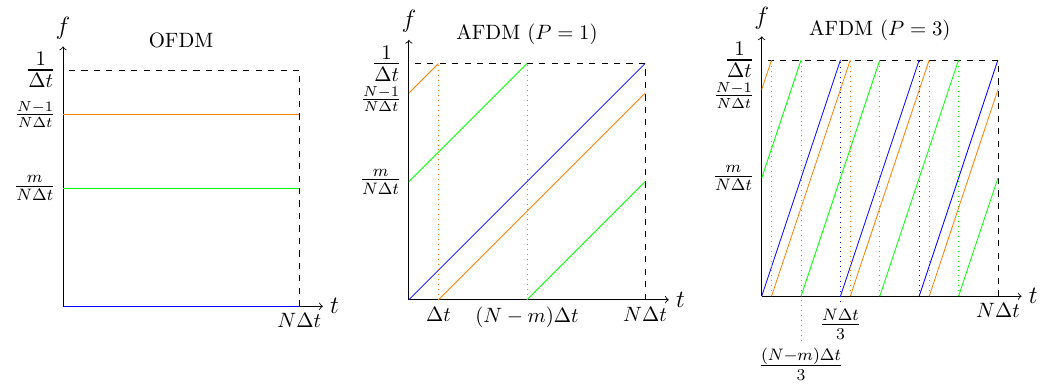}
    \caption{Time-frequency representation of three subcarriers of OFDM and AFDM ($c_1=\frac{P}{2N}$). Each subcarrier is represented with a different colour.}
    \label{fig:afdm_t_f}
\end{figure}

\section{Background: Hierarchical sparsity}
\begin{definition}[Hierarchical sparsity, \cite{hierarchical}]
\label{def:hierarchical}
A vector $\mathbf{x} \in \mathbb{C}^{NM}$ is $\left(s_{N},s_{M}\right)$-sparse if $\mathbf{x}$ consists of $N$ blocks each of size $M$, with at most $s_{N}$ blocks having non-vanishing elements and each non-zero block itself being $s_{M}$-sparse.
\end{definition}
To analyze hierarchically sparse recovery schemes, a modified version of the restricted isometry property (RIP) called the hierarchical RIP (HiRIP) was proposed in the literature. 
\begin{definition}[HiRIP, \cite{hierarchical}]
\label{def:HiRIP}
The HiRIP constant of a matrix $\mathbf{A}$, denoted by $\delta_{s_{N},s_{M}}$, is the smallest $\delta \geq 0$ such that 
\begin{equation}
     \label{eq:HiRIP}
     (1- \delta)\left\|\mathbf{x}\right\|^2\leq \left\|\mathbf{A}\mathbf{x}\right\|^2\leq(1+ \delta)\left\|\mathbf{x}\right\|^2
 \end{equation}
 for all $\left(s_{N},s_{M}\right)$-sparse vectors $\mathbf{x} \in \mathbb{C}^{NM}$.
\end{definition}

\section{System model}
\subsection{Doubly sparse linear time-varying (DS-LTV) channels}
In an LTV channel with $L$ paths, the complex gain $h_{l,n}$ ($n\in\Iintv{-L_{\rm CPP},N-1}$) of the $l$-th path varies with time as
 \begin{equation}
     \label{eq:ch_model}
     h_{l,n}=\sum_{q=-Q}^{Q}\alpha_{l,q}I_{l,q}e^{\imath 2\pi\frac{nq}{N}}, l=0\cdots L-1
 \end{equation}
Here, $I_{l,q}$ for any $l$ and $q$ is a binary random variable that, when non-zero, indicates that a channel path with delay $l$, Doppler shift $q$ and complex gain $\alpha_{l,q}$ is active and contributes to the channel output. Note that the distribution of the random variables $\left\{I_{l,q}\right\}_{l,q}$ controls the kind of sparsity the LTV channel might have. The complex gain is assumed to satisfy $\alpha_{l,q}\sim\mathcal{CN}\left(0,\sigma_{\alpha}^2\right)$ with $\sigma_{\alpha}^2$ satisfying the channel power normalization $\sum_{l=0}^{L-1}\sum_{q=-Q}^{Q}\mathbb{E}\left[\left|\alpha_{l,q}\right|^2I_{l,q}\right]=1$. Note that this model is an on-grid approximation of a time-varying channel. For instance, the Doppler shifts are assumed to be integers in $\Iintv{-Q,Q}$ when normalized with the resolution associated with the transmission duration.
\begin{definition}[Delay-Doppler double sparsity, \cite{afdm_gc}]
\label{def:dD_sparsity}
An LTV channel is doubly sparse if there exist $0<p_d,p_D<1$ s.t.
\begin{equation}
    I_{l,q}=I_lI_q^{(l)}, \forall (l,q)\in\Iintv{0,L-1}\times\Iintv{-Q,Q}
\end{equation}
where $I_l\sim\mathrm{Bernoulli}(p_d)$ and $I_q^{(l)}\sim\mathrm{Bernoulli}(p_D)$.
\end{definition}
Note that under Definition \ref{def:dD_sparsity}, $s_{\rm d}\triangleq\mathbb{E}\left[\sum_{l}I_l\right]=p_dL$ is the mean number of active delay taps in the delay-Doppler profile of the channel and can be thought of as the {\it delay domain sparsity level} while $s_{\rm D}\triangleq\mathbb{E}\left[\sum_{q}I_q^{(l)}\right]=p_D(2Q+1)$ is the mean number of active Doppler bins per delay tap and can be thought of as the {\it Doppler domain sparsity level}.
Fig. \ref{fig:examples} illustrates three different delay-Doppler sparsity models, fully described in \cite{afdm_gc} and dubbed Type-1, Type-2 and Type-3, that all fall under the scope of Definition \ref{def:dD_sparsity} each with an additional different assumption on $I_l$ and $I_q^{(l)}$. Here, we just point out that the difference between Type-2 and Type-3 of Figures \ref{fig:examples}-(b) and \ref{fig:examples}-(c), respectively, is that in the latter the active Doppler bins per delay tap appear in clusters of random positions but of deterministic length as opposed to the absence of clusters in the former. The case where all the delay taps have the same (random) sparsity (as in Type-1 models of Fig. \ref{fig:examples}-(a)) also falls under Definition \ref{def:dD_sparsity} by setting $I_q^{(l)}=I_q^{(0)},\forall l$.
\begin{figure}
  \centering
  \begin{tabular}{ c @{\hspace{5pt}} c @{\hspace{5pt}} c}
  \includegraphics[width=.3\columnwidth] {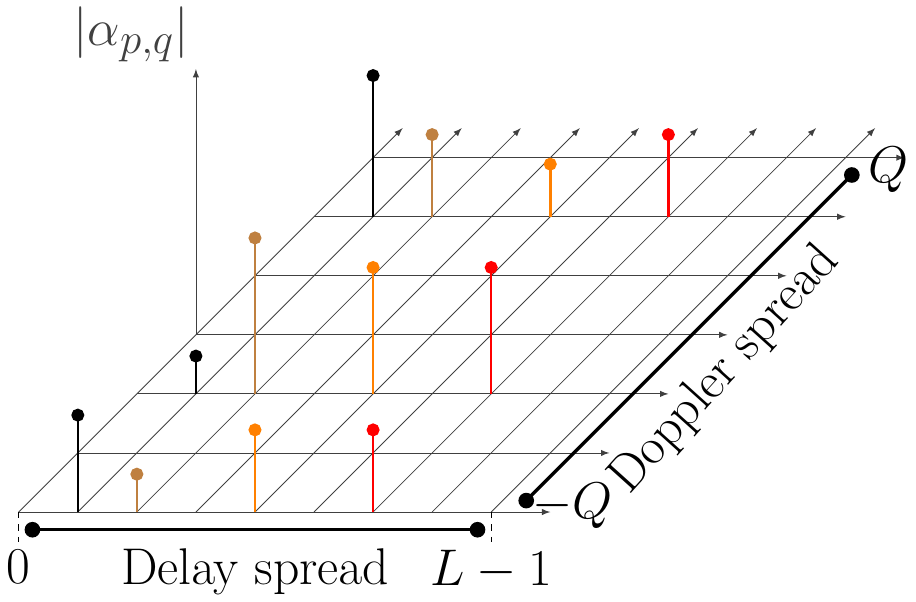} &
    \includegraphics[width=.3\columnwidth]{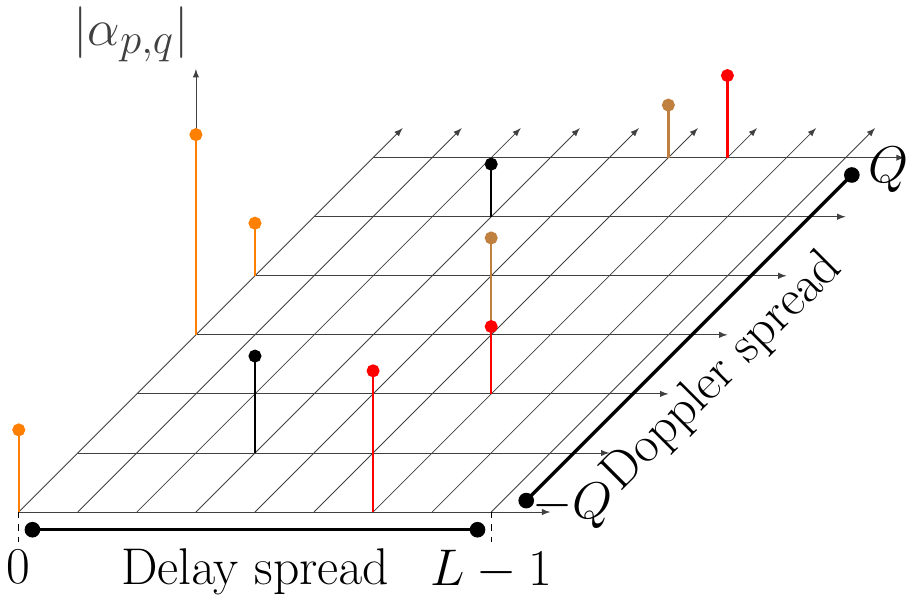} &
      \includegraphics[width=.3\columnwidth]{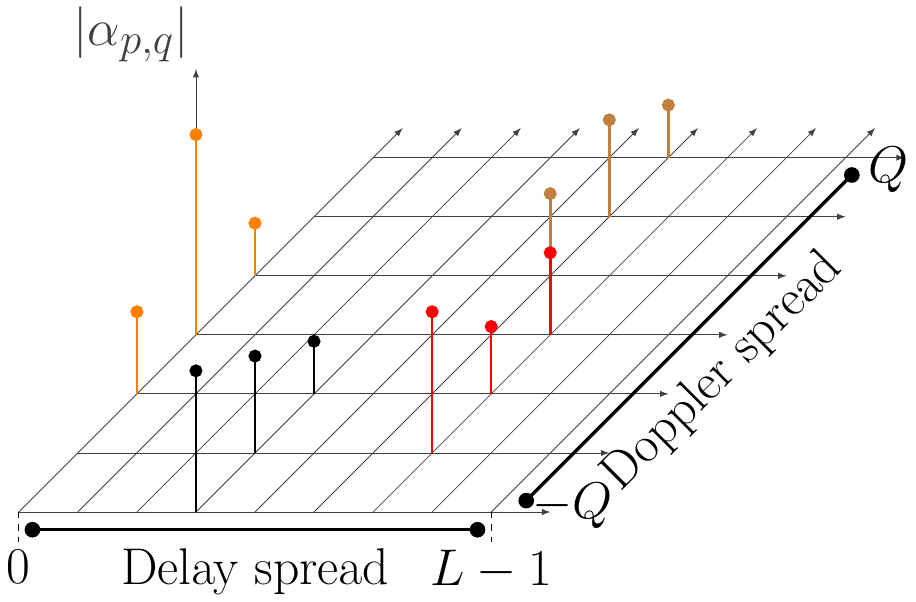} \\
    \small (a) &
      \small (b)&
      \small (c)
  \end{tabular}
  \medskip
  \caption{Examples of channels satisfying (a) Type-1 delay-Doppler sparsity, (b) Type-2 delay-Doppler sparsity, (c) Type-3 delay-Doppler sparsity}
  \label{fig:examples}
\end{figure}

\subsection{Relation to hierarchical sparsity}
DS-LTV sparsity is probabilistic while hierarchical sparsity of Definition \ref{def:hierarchical} is deterministic. The two models are nonetheless related: if vectorized to a concatenation of its rows, random matrix $\left[\alpha_{l,q}I_{l,q}\right]_{l,q}$ defines a vector $\boldsymbol{\alpha}\in\mathbb{C}^{(2Q+1)L}$ that consists of $L$ blocks each of size $2Q+1$, where in average $s_{\rm d}$ blocks have non-zero entries and where each non-zero block itself is in average  $s_{\rm D}$-sparse. To ensure sparsity in a stronger sense i.e., with high probability (as $L,Q,Lp_{\rm d},(2Q+1)p_{\rm D}$ grow), we require that the following assumption hold.
\begin{assumption}
    \label{assum:technical}
    $\left\{I_l\right\}_{l=0\cdots L-1}$ are mutually independent. Moreover, the complementary cumulative distribution function (CCDF) $\overline{F}_{S_{\mathrm{D},l}}(m)$ of the random variable $S_{\mathrm{D},l}\triangleq\sum_{q=-Q}^{Q}I_{q}^{(l)}$ for any $l\in\Iintv{0,L-1}$ is upper-bounded for any integer $m>(2Q+1)p_{\rm D}$ by the CCDF of $\mathrm{B}\left(2Q+1,p_{\rm D}\right)$.
\end{assumption}
Type-1 and 2 models are made to satisfy the CCDF upper bound by requiring that $\{I_{q}^{(0)}\}_q$ in the first and $\{I_{q}^{(l)}\}_q$ for any $l$ in the second to be mutually independent (and to thus satisfy $\overline{F}_{S_{\mathrm{D},l}}(m)=\overline{F}_{\mathrm{B}\left(2Q+1,p_{\rm D}\right)}(m),\forall m$). For Type-3 models, $S_{\mathrm{D},l}$ is deterministic and hence its CCDF is trivially upper-bounded.
As the following lemma rigorously shows, the mutual independence of $\left\{I_l\right\}_{l=0\cdots L-1}$ in Assumption \ref{assum:technical} guarantees strong delay domain sparsity while Doppler sparsity is guaranteed in a more explicit manner by the CCDF upper bound.
\begin{lemma}
 \label{lem:DS_HS}
 Under Assumption \ref{assum:technical}, the vector $\boldsymbol{\alpha}$ is $\left(s_{\rm d},s_{\rm D}\right)$-sparse with probability $1-e^{-\Omega\left(\min\left((2Q+1)p_{\rm D},Lp_{\rm d}\right)\right)}$.
\end{lemma}
\begin{proof}
    The proof of the lemma is given in Appendix \ref{app:proof_lem}.
\end{proof}

\subsection{AFDM signal model on DS-LTV channels}
The received samples at the channel output are
\begin{equation}
    r_n = \textstyle\sum_{l = 0}^{L-1}s_{n-l}h_{l,n}\label{r_n} + z_n,\quad n= 0\cdots N-1,
\end{equation}
where $z_n\sim\mathcal{CN}\left(0,\sigma_w^2\right)$ represents the i.i.d. Gaussian noise process. After discarding the CPP (assumed to satisfy $L-1\leq L_{\rm CPP}$), the DAFT domain output symbols are
\begin{align}
    &y_k = \frac{1}{\sqrt{N}}\sum_{n = 0}^{N-1}r_ne^{-\imath2\pi (c_2k^2+{\frac{kn}{N}}+c_1n^2)}, k= 0\cdots N-1\nonumber\\
    &=\sum_{l = 0}^{L-1}\sum_{q=-Q}^{Q}\alpha_{l,q}I_{l,q} 
    e^{\imath2\pi(c_1l^2-\frac{ml}{N} + c_2(m^2 - k^2))}x_m +w_k,
    \label{eq:y_output_integer}
\end{align}
where the second equality is obtained using the input-output relation given in \cite{BemaniAFDM_TWC}, $w_k$ is i.i.d. and $\sim\mathcal{CN}\left(0,\sigma_w^2\right)$ and where $ m \triangleq (k - q + 2Nc_1 l)_N$. Note how the Doppler components of different delay taps are mixed in the DAFT domain since a path occupying the $(l,q)$ grid point in the delay-Doppler domain appears as a $q-2Nc_1l$ shift in the DAFT domain.

\section{Compressed-sensing estimation of DS-LTV channels using AFDM}
\subsection{DS-LTV compressed-sensing channel estimation problem}
Let $\mathcal{P}\subset\Iintv{0,N-1}$ designate the indexes of the $N_{\rm p}\left(2|c_1|N(L-1)+2Q+1\right)$ received samples associated with $N_{\rm p}$ DAFT domain pilots, of values $\{{\rm p}_p\}_{p=1\cdots N_{\rm p}}$ inserted at indexes $\left\{m_p\right\}_{p=1\cdots N_{\rm p}}$ so as each pilot is preceded by $Q$ zero samples and followed by $\left(2|c_1|N(L-1)+Q\right)$ zero samples\footnote{We show in Appendix \ref{app:proof_theo} that the recovery results we prove hold even if we reduce the cardinality of $\mathcal{P}$ to $N_{\rm p}\left(2|c_1|N(L-1)+1\right)+2Q$ e.g., by allowing partial overlapping between neighbouring pilot guard intervals.}. Vector $\mathbf{y}_{\rm p}\triangleq[y_k]_{k\in\mathcal{P}}$ is the vectorized form of the received pilot samples. Referring to \eqref{eq:y_output_integer}, we can write
\begin{equation}
     \label{eq:yp}
     \mathbf{y}_{\rm p}=\underbrace{\mathbf{A}_{\cal P}\mathbf{M}}_{\triangleq\mathbf{M}_{\rm p}}\boldsymbol{\alpha}+\mathbf{w}_{\rm p}
 \end{equation}
where $[\mathbf{M}]_{l(2Q+1)+Q+q+1}=\boldsymbol{\Phi}\boldsymbol{\Delta}_q \boldsymbol{\Pi}^l \boldsymbol{\Phi}^H \mathbf{x}_{\rm p}$, $\mathbf{x}_{\rm p}$ is a $N$-long vector with entries equal to ${\rm p}_1, \ldots, {\rm p}_{N_{\rm p}}$ at indexes $\left\{m_p\right\}_{p=1\cdots N_{\rm p}}$ and to zero elsewhere, and $\mathbf{w}_{\rm p}\triangleq[\Tilde{{w}}]_{k\in\mathcal{P}}$. Here, $\mathbf{A}_{\cal P}$ is the $|\mathcal{P}|\times N$ matrix that chooses from a $N$-long vector the entries corresponding to $\mathcal{P}$.
 $\boldsymbol{\Delta}_q=\mathrm{diag}(e^{\imath2\pi qn},n=0\cdots N-1)$, $\boldsymbol{\Pi}$ is the $N$-order permutation matrix, $\boldsymbol{\Phi}=\pmb{\Lambda}_{c2} \mathbf{F}_N \pmb{\Lambda}_{c1}$ with $\mathbf{F}_N$ being the $N$-order discrete Fourier transform (DFT) matrix and $\boldsymbol{\Lambda}_{c}=\mathrm{diag}(e^{-\imath2\pi cn^2},n=0\cdots N-1)$.
 Recall that $\boldsymbol{\alpha}$ is hierarchically sparse due to Lemma \ref{lem:DS_HS}. Its sparsity support is assumed to be unknown to the receiver.
 %
\subsection{Algorithms for compressed sensing of DS-LTV channels}
The hierarchical hard thresholding pursuit (HiHTP) approach has been suggested in the literature \cite{hierarchical} for solving  hierarchically-sparse recovery problems such as Problem \eqref{eq:yp} for which it gives Algorithm \ref{algo:HiHTP}.
\begin{algorithm} 

    \caption{HiHTP for AFDM based compressed sensing}
    \label{algo:HiHTP}
    \begin{algorithmic}[1]
        \STATE \textbf{Input:} $\mathbf{M}_{\rm p}$, $\mathbf{y}_{\rm p}$, maximum number of iterations $k_{\max}$,  $s_{\rm d}$, $s_{\rm D}$
        \STATE $\hat{\boldsymbol{\alpha}}^{\left(0\right) }= 0$, $k=0$
        \STATE \textbf{repeat}
        \STATE $\Omega^{(k+1)}=L_{s_{\rm d},s_{\rm D}}\left(\boldsymbol{\alpha}^{\left(k\right)}+\mathbf{M}_{\rm p}^{\rm H}\left(\mathbf{y}_{\rm p}-\mathbf{M}_{\rm p}\boldsymbol{\alpha}^{\left(k\right)}\right)\right)$
        \STATE $\boldsymbol{\alpha}^{\left(k+1\right)}=\underset{\mathbf{z}\in\mathbb{C}^{L(2Q+1)}}{\arg \min}\left\{\left\|\mathbf{y}_{\rm p}-\mathbf{M}_{\rm p}\mathbf{z}\right\|,\sup\left(\mathbf{z}\right)\subset\Omega^{(k+1)}\right\}$
        \STATE $k=k+1$
        \STATE \textbf{until} $k=k_{\max}$ or $\Omega^{(k+1)}=\Omega^{(k)}$ (whichever earlier)
        \STATE \textbf{Output:} $\left(s_{\rm d},s_{\rm D}\right)$-sparse $\hat{\boldsymbol{\alpha}}^{\left(k\right)}$.
        
    \end{algorithmic}
\end{algorithm}
HiHTP is a modification of the classical hard thresholding pursuit (HTP) \cite{Foucart_HTP} by replacing the thresholding operator employed at each iteration of HTP with a \emph{hierarchically} sparse version $L_{s_{\rm d},s_{\rm D}}$.
To compute $L_{s_{\rm d},s_{\rm D}}(\mathbf{x})$ for a vector $\mathbf{x}\in\mathbb{C}^{L(2Q+1)}$ first a $s_{\rm D}$-sparse approximation is applied to each one of the $L$ blocks of $\mathbf{x}$ by keeping in each of them the largest $s_{\rm D}$ entries while setting the remaining ones to zero. A $s_{\rm d}$-sparse approximation is next applied to the result by identifying the $s_{\rm d}$ blocks with the largest $l_2$-norm. 

\subsection{Analyzing AFDM compressed sensing of DS-LTV channels}
To guarantee the convergence of Algorithm \ref{algo:HiHTP} and the recovery of $\boldsymbol{\alpha}$, the following technical assumption is needed.
\begin{assumption}
    \label{assum:technical_diag}
    Random variables $\{I_{q}^{(l_1)}\}_{q=-Q\cdots Q}$ are independent from $\{I_{q}^{(l_2)}\}_{q=-Q\cdots Q}$ for any $l_1\neq l_2$.
\end{assumption}
\begin{theorem}[HiRIP for AFDM based measurements]
\label{theo:HiRIP_AFDM}
Assume $\left|c_1\right|=\frac{P}{2N}$ and $P$ is set as the \emph{smallest} integer satisfying $(L-1)P+2Q+1\geq s_{\rm d}s_{\rm D}$. Then under Assumptions \ref{assum:technical} and \ref{assum:technical_diag} and for sufficiently large $L$, $Q$, sufficiently small $\delta$, and $N_{\rm p}>O\left(\frac{1}{\delta^{2}}\log^2\frac{1}{\delta}\log\frac{\log (LP)}{\delta}\log(LP)\log\frac{Q}{P}\right)$, the HiRIP constant $\delta_{s_{\rm d},s_{\rm D}}$ of matrix $\mathbf{M}_{\rm p}$ satisfies $\delta_{s_{\rm d},s_{\rm D}}\leq\delta$ with probability $1-e^{-\Omega \left(\log{\left(2\lceil\frac{Q}{P}\rceil+1\right)}\log{\frac{1}{\delta}}\right)}$. 
\end{theorem}
When $P=2Q+1$ AFDM achieves full diversity \cite{BemaniAFDM_TWC} and the measurements are non-compressive, while $P=1$ is the most compressive. By setting $P$ as in the theorem between these two extremes, each pilot instance gives in its $(L-1)P+2Q+1$-long guard interval a number of measurements close with high probability to the number $s_{\rm d}s_{\rm D}$ of unknowns. Of course, a number $N_{\rm p}>1$ of pilot instances is still required as the sparsity support needs to be estimated. But, asymptotically, this number has only a logarithmic growth.
\begin{proof}
The outlines of the proof is given in Appendix \ref{app:proof_theo}.
\end{proof}
\begin{corollary}[Recovery guarantee for AFDM based measurements]
\label{cor:HiHTP_convergence}
The sequence $\hat{\boldsymbol{\alpha}}^{(k)}$ defined by Algorithm \ref{algo:HiHTP} satisfies
\begin{equation}
    \|\hat{\boldsymbol{\alpha}}^{(k)}-\boldsymbol{\alpha}\|\leq \rho^k\|\boldsymbol{\alpha}^{(0)}-\boldsymbol{\alpha}\|+\tau\left\|\mathbf{w}_{\rm p}\right\|
\end{equation}
where $\rho<1$ and $\tau$ are constants defined in \cite[Theorem 1]{hierarchical}.
\end{corollary}
\begin{proof}
Thanks to Theorem \ref{theo:HiRIP_AFDM}, matrix $\mathbf{M}_{\rm p}$ with large-enough $L,Q,N_{\rm p}$ has a HiRIP constant that satisfies $\delta_{3s_{\rm d},2s_{\rm D}}<\frac{1}{\sqrt{3}}$. The conditions of \cite[Theorem 1]{hierarchical} are thus satisfied and the corollary follows from that theorem.
\end{proof}
We next explain how the value of $N_{\rm p}$ dictated by Theorem \ref{theo:HiRIP_AFDM} translates into sub-Nyquist sampling rates for radar receivers.
\section{Application to sub-Nyquist radar}
\label{sec:sub-nyquist}
We now consider the case where the AFDM signal is destined for a sensing receiver either co-located with the transmitter (the mono-static setting) or in a remote device (the bi-static setting). In any of these settings, the non-zero complex gains $\alpha_{l,q}$ in \eqref{eq:ch_model} will represent a point target with a delay $l$ (related to the to-be-estimated range) and a Doppler frequency shift $q$ (related to the to-be-estimated velocity).

Instead of applying DAFT to the received AFDM signal after sampling as in basic AFDM operation \cite{BemaniAFDM_TWC} (which would require a sampling rate at least equal to the signal bandwidth), an alternative consists in first de-chirping the received signal in the analog domain  with a continuous-time version \cite{AFDM_letter} of a DAFT chirp carrier e.g., of the $0$-th chirp $\left(e^{\imath2\pi (c_20^2+{\frac{1}{N}}0n+c_1n^2)}\right)_{n}$. The result is a multi-tone signal (as shown in Fig. \ref{fig:SamplingRate} in the case of $N_{\rm p}=1$ and $P=2$) with discontinuities due to the frequency wrapping characterizing AFDM chirp carriers. In this figure, the de-chirped signal occupies two disjoint frequency bands that get merged into one (without discontinuities) thanks to spectrum folding after sampling at rate $f_{\rm s}=\frac{(L-1)P+1}{T}$.
\begin{figure}
    \centering
    \includegraphics[width=0.9\columnwidth]{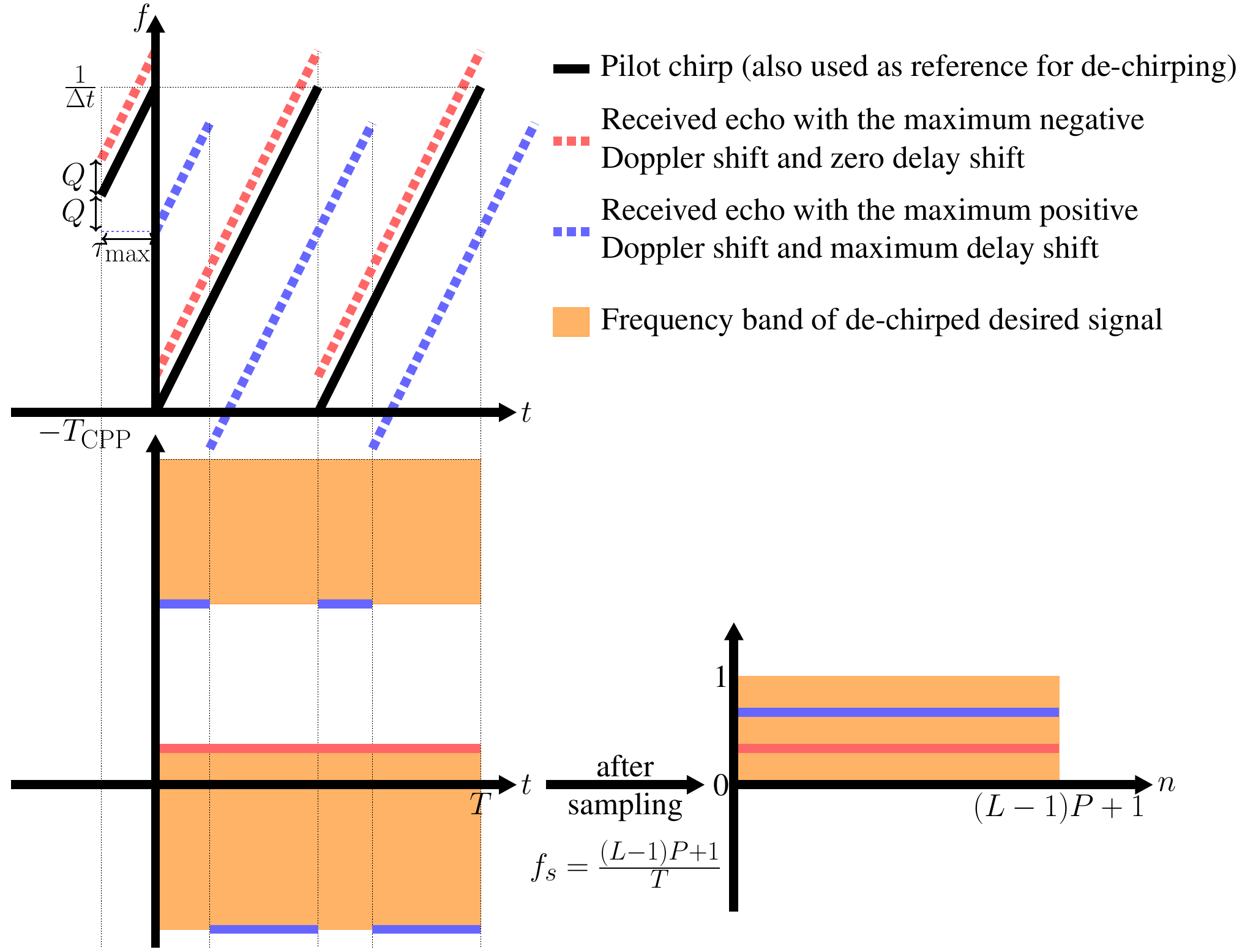}
    \caption{Time-frequency content of one AFDM pilot and its echoes, before and after analog de-chirping and sampling ($\tau_{\max} \triangleq (L-1)\Delta t$, $T_{\rm CPP} \triangleq L_{\rm CPP}\Delta t$)}
    \label{fig:SamplingRate}
\end{figure}
In the general case of $N_{\rm p}\geq1$ pilots, if we restrict the total subset $\mathcal{P}$ of pilot guard indexes to be an interval, then sampling after de-chirping can be done at rate $f_{\rm s}=\frac{N_{\rm p}\left((L-1)P+1\right)}{T}$ to yield the vector $\mathbf{y}_{\rm p}$ used for target estimation. In most practical configurations $\frac{N_{\rm p}\left((L-1)P+1\right)}{T}\ll\frac{N}{T}=\frac{1}{\Delta t}$, and hence the sampling rate needed for AFDM sensing is significantly smaller than what is needed in sensing based on OFDM or OTFS waveforms.
\section{Numerical results}
\label{sec:simulations}
AFDM sparse recovery performance is now compared to that of OFDM and OTFS. For OFDM, transmission is organized in $N$-long frames, each constructed from $N_{\rm ofdm,symb}\approx2Q+1$ OFDM symbols each of which costing $L-1$ in CP overhead. Within each frame, $N_{\rm p,fd}$ subcarriers within $N_{\rm p,td}$ OFDM symbols are set as pilots \cite{afdm_gc}. As for OTFS, subcarriers are in the delay-Doppler domain forming a $M_{\rm otfs}\times N_{\rm otfs}$ grid (with $M_{\rm otfs}N_{\rm otfs}=N$). OTFS with orthogonal data-pilot resources \cite{embedded_otfs_bem} requires at least $N_{\rm p,otfs}=1$ pilot symbols with $\min(4Q+1,N_{\rm otfs})\min(2L-1,M_{\rm otfs})$ guard samples.

We used 100 realizations of channels having a Type-1 delay-Doppler sparsity with $p_{\rm d}=0.2$, $p_{\rm D}\in\{0.2,0.4\}$ and $N=4096, L=30, Q=7$ (corresponding to a 30 MHz transmission at a 70 GHz carrier frequency, a maximum target moving speed of 396 km/h and a maximum target range of 300 meters). For both AFDM and OFDM, sparse recovery of $\boldsymbol{\alpha}$ is done using HiHTP (Algorithm \ref{algo:HiHTP}). For OTFS, since sensing is done without compression, non-compressive estimation algorithms can be used \cite{BemaniAFDM_TWC}.
For each waveform, the number of pilots was set in such a way that the mean squared error ${\rm MSE}\triangleq\mathbb{E}[\left\|\hat{\boldsymbol{\alpha}}-\boldsymbol{\alpha}\right\|^2]$ is approximately $10^{-4}$ at $\mathrm{SNR}=20$ dB. Fig. \ref{fig:simulations1_2} shows an advantage of AFDM in terms of pilot overhead i.e., the number of samples in each frame needed as pilots and guards to achieve the target MSE performance. 
   \begin{figure}
  \centering
  \begin{tabular}{ c }
  \includegraphics[width=.77\columnwidth]{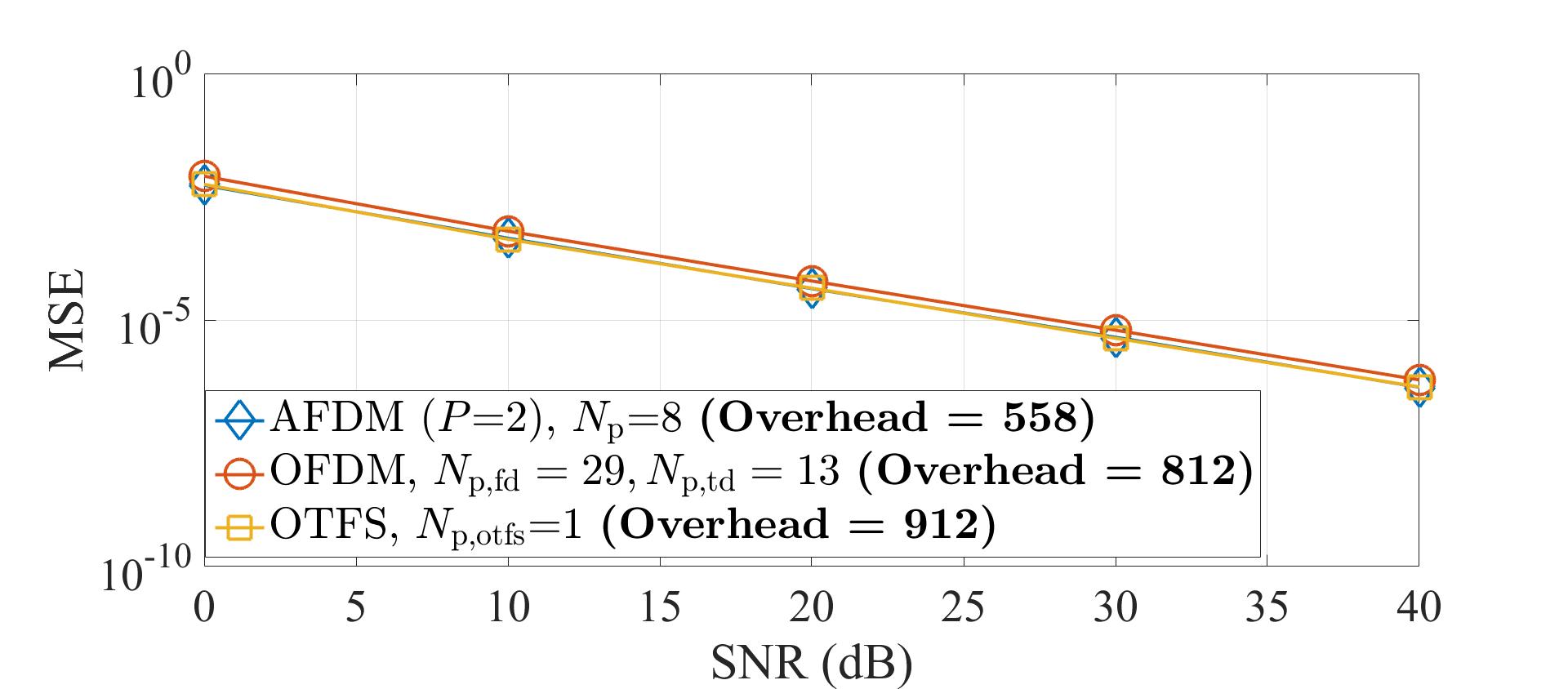} \\
  \small (a) $p_{\rm D}=0.2$\\
    \includegraphics[width=.77\columnwidth]{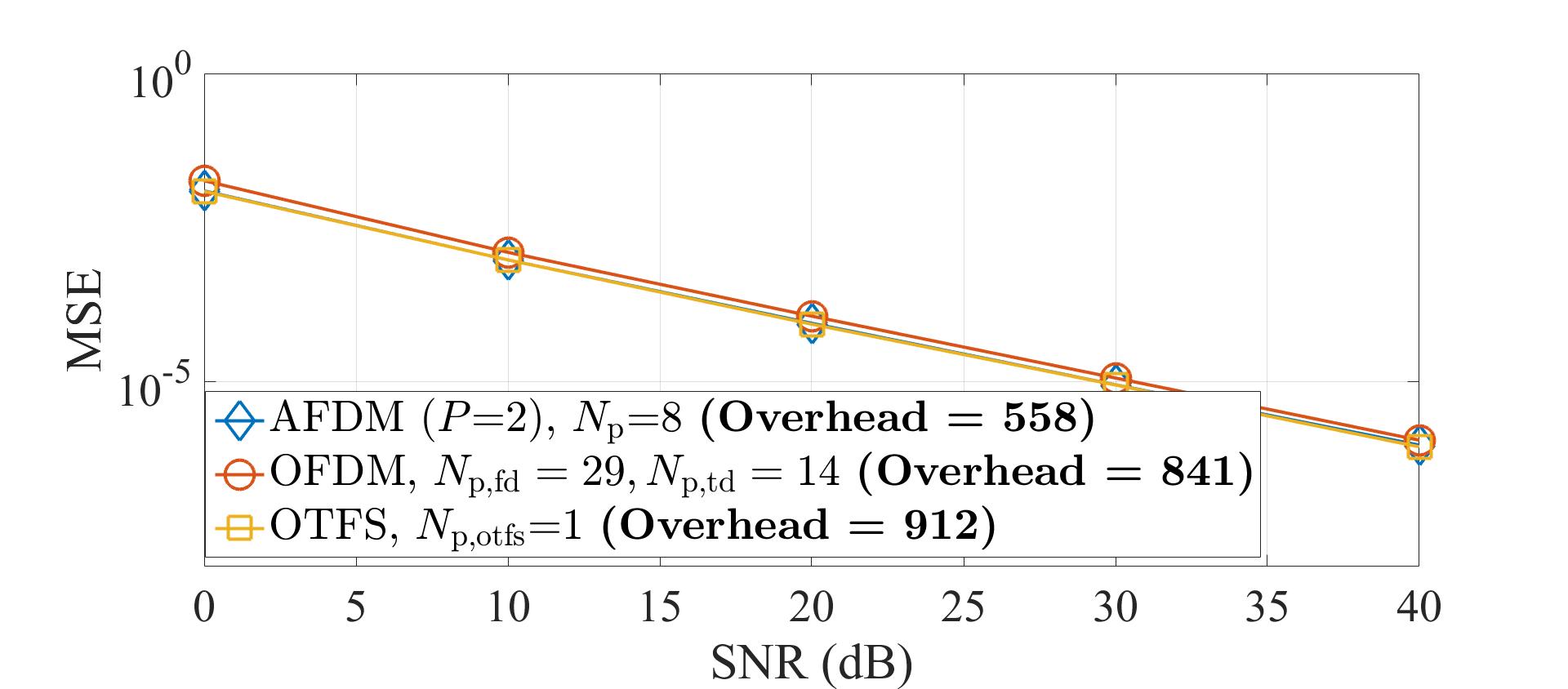} \\
      \small (b) $p_{\rm D}=0.4$
  \end{tabular}
  \medskip
  \caption{MSE and pilot overhead for $N=4096, L=30, Q=7, p_{\rm d}=0.2$, $N_{\rm ofdm,symb}=16$, $N_{\rm otfs}=16$, $M_{\rm otfs}=256$. Overhead: $N_{\rm p,td}N_{\rm p,fd}+(N_{\rm ofdm,symb}-1)(L-1)$ for OFDM, $\min(4Q+1,N_{\rm otfs})\min(2L-1,M_{\rm otfs})$ for OTFS, $N_{\rm p}\left((L-1)P+1\right)+(L-1)P+4Q$ for AFDM.}
  \label{fig:simulations1_2}
  \vspace{-4mm}
\end{figure}
However, the main focus in this paper is the gain that can be achieved in terms of sampling rate reduction when AFDM is employed for sub-Nyquist sensing. This gain is illustrated (for the same setting as Fig. \ref{fig:simulations1_2}) by Table \ref{tab:table1}.
\begin{table}[h!]
\vspace{-2mm}
  \begin{center}
    \caption{Minimal sampling rate at sensing receiver}
    \label{tab:table1}
    \begin{tabular}{m{0.18\linewidth} m{0.14\linewidth} m{0.14\linewidth} m{0.32\linewidth}}
    
 Waveform&OFDM & OTFS & AFDM \\
 \hline
   Sampling rate $f_{\rm s} $ (MHz)
& $30$ $=\mathrm{BW}$ &$30$ $=\mathrm{BW}$& { $3.45$ $=\frac{N_{\rm p}\left((L-1)P+1\right)}{T}$}\\
\hline

    \end{tabular}
    \end{center}
    \vspace{-2mm}
    \end{table}

\section{Conclusions}
The advantage of using AFDM instead of measurement matrices based on other waveforms for sub-Nyquist sensing to recover doubly-sparse delay-Doppler profiles has been rigorously established by linking delay-Doppler sparsity to the paradigm of hierarchically-sparse recovery. Future work will address the problem without any on-grid approximation.

\appendices
\section{Proof of Lemma \ref{lem:DS_HS}}
\label{app:proof_lem}
Let $S_{\rm d}\triangleq\sum_{l=0}^{L-1}I_l$ be the number of active delay taps. From Definition \ref{def:dD_sparsity} and Assumption \ref{assum:technical}, $S_{\rm d}\sim\mathrm{B}\left(L,p_{\rm d}\right)$. applying the Chernoff's bound to $S_{\rm d}$ evaluated at $s_{\rm d}=(1+\epsilon)Lp_{\rm d}$ (with an $\epsilon>0$ that can be set as small as needed) gives
\begin{align}
    \label{eq:chernoff_S_d}
    \mathbb{P}\left[S_{\rm d}>s_{\rm d}\right]&\leq\left(\frac{p_{\rm d}}{\frac{s_{\rm d}}{L}}\right)^{s_{\rm d}}\left(\frac{1-p_{\rm d}}{1-\frac{s_{\rm d}}{L}}\right)^{L-s_{\rm d}}=e^{-\Omega(Lp_{\rm d})}.
\end{align}
As for $S_{\mathrm{D},l}$, since Assumption \ref{assum:technical} upper-bounds its CCDF by that of a $\mathrm{B}\left(2Q+1,p_{\rm D}\right)$ distribution, applying the Chernoff's bound to the latter evaluated at $s_{\rm D}=(1+\epsilon)(2Q+1)p_{\rm D}$ similarly gives \emph{joint} sparsity of $\{I_{q}^{(l)}\}_{l=0\cdots L-1}$ in the sense
\begin{equation}
    \label{eq:chernoff_S_D}
    \mathbb{P}\left[\exists l, I_l=1,S_{\mathrm{D},l}>s_{\rm D}\right]=e^{-\Omega((2Q+1)p_{\rm D})}.
\end{equation}
Combining \eqref{eq:chernoff_S_d} and \eqref{eq:chernoff_S_D} completes the proof of the lemma.

\section{Outlines of the proof of Theorem \ref{theo:HiRIP_AFDM}}
\label{app:proof_theo}
First, for each $l\in\Iintv{0,(L-1)P}$ we define $\mathcal{D}_l\triangleq\left\{(\tilde{l},q) \mathrm{s.t.}(q+P\tilde{l})_{(L-1)P+1}=l\right\}$ as the set of delay-Doppler grid points that potentially contribute to the pilot sample received at DAFT domain index $l$ (Fig. \ref{fig:interval_k}). Next, we define
$\Tilde{\boldsymbol{\alpha}}\triangleq\left[\boldsymbol{\alpha}_{\mathcal{D}_0}^{\rm T}\quad\cdots\quad\boldsymbol{\alpha}_{\mathcal{D}_{(L-1)P}}^{\rm T}\right]^{\rm T}$ where $\boldsymbol{\alpha}_{\mathcal{D}_l}\triangleq\left[\alpha_{l,q}\right]_{(l,q)\in\mathcal{D}_l}$. The entries of $\Tilde{\boldsymbol{\alpha}}$ are just a permutation of the entries of $\boldsymbol{\alpha}$ and estimating one of them directly gives an estimate of the other. Next, it can be shown that when $P$ is set as in the theorem $\Tilde{\boldsymbol{\alpha}}$ is $\left(\tilde{s}_{\rm d},\tilde{s}_{\rm D}\right)$-hierarchically sparse with high probability where
\begin{figure}
  \centering
  \begin{tabular}{ c @{\hspace{5pt}} c }
  \includegraphics[width=.4\columnwidth]{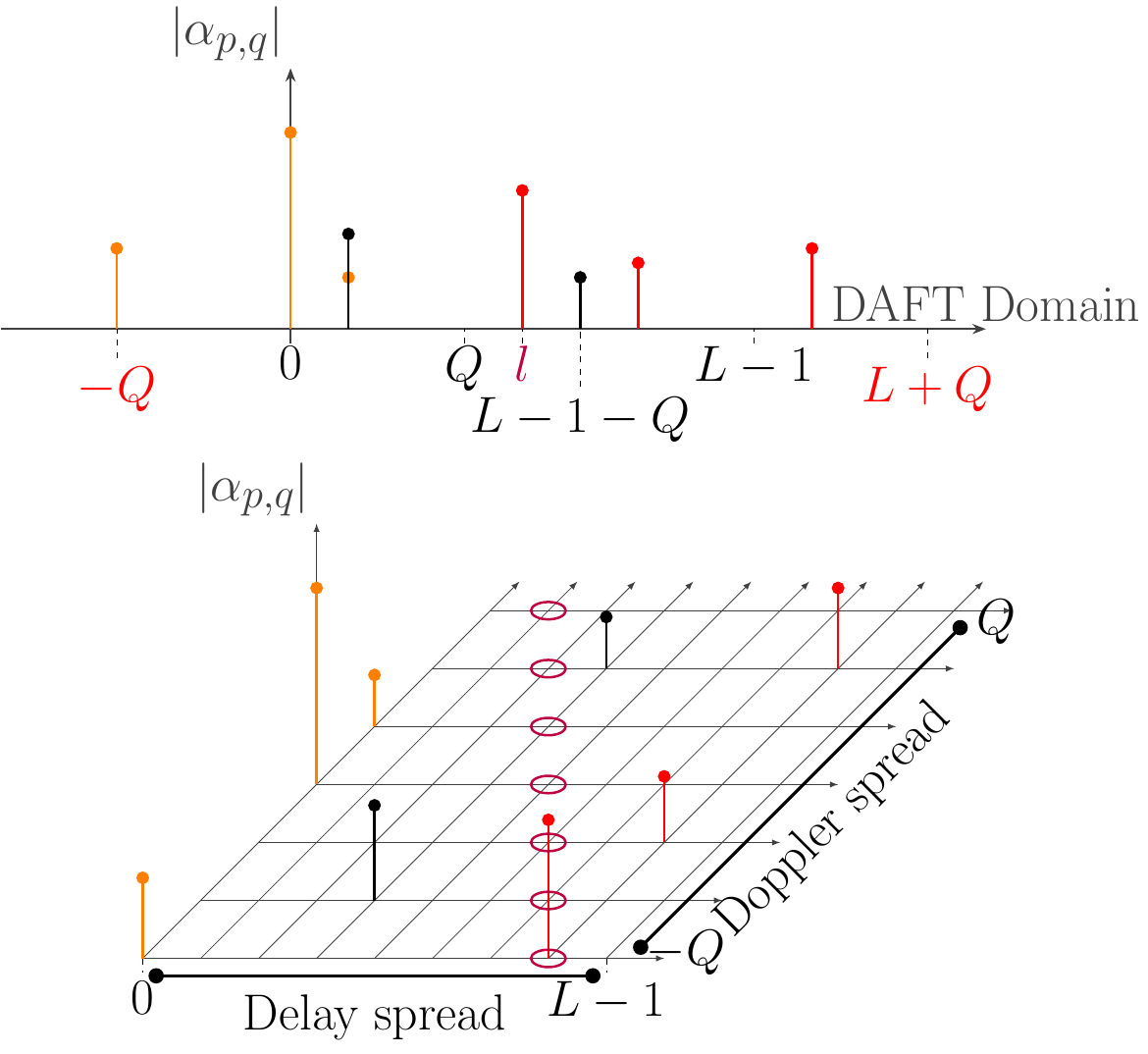} &
    \includegraphics[width=.4\columnwidth]{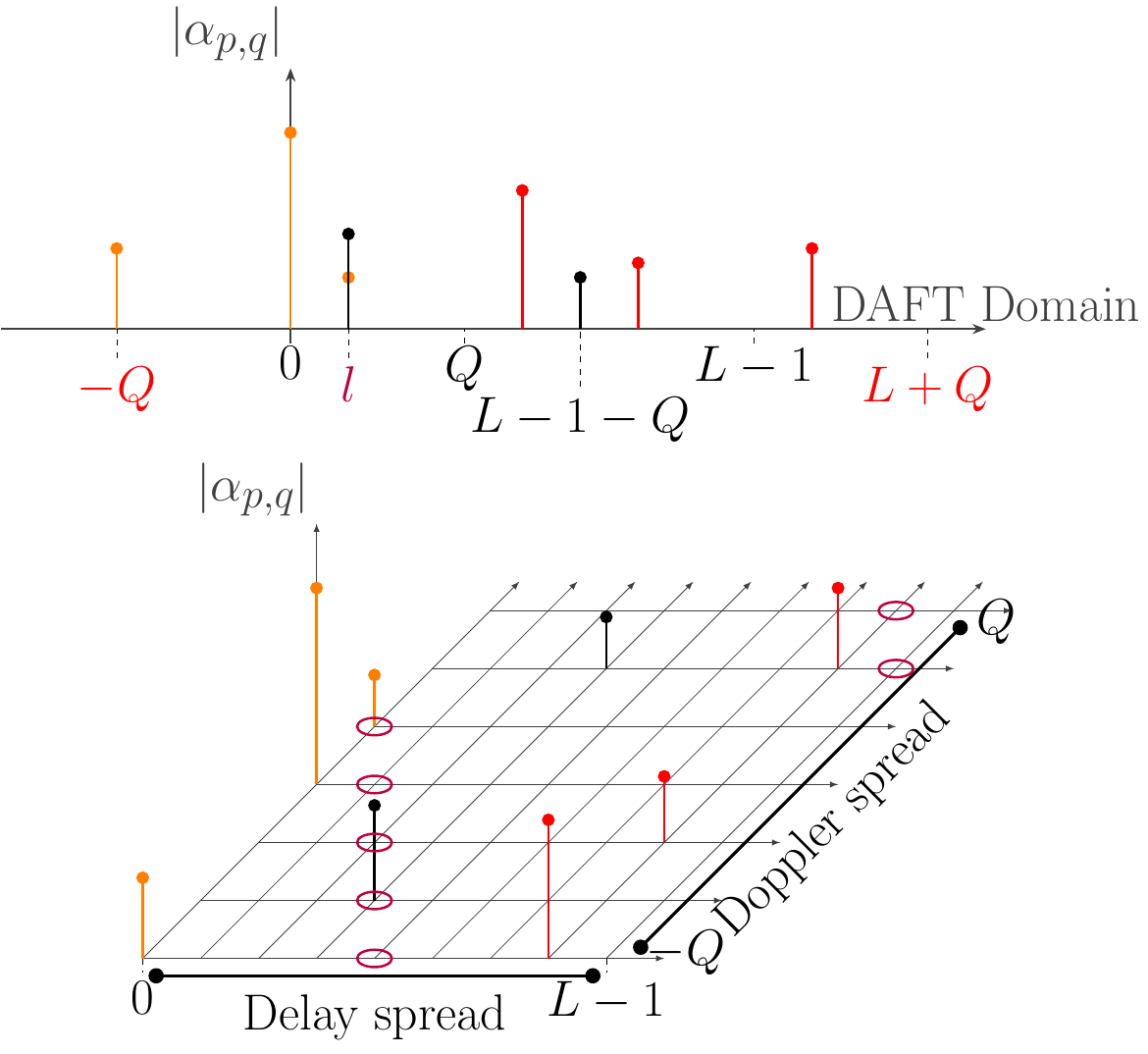} \\
    \small (a) &
      \small (b)
  \end{tabular}
  \medskip
  \vspace{-2mm}
  \caption{Two examples of the set $\mathcal{D}_l$ ($P=1$) (a) for an $l$ resulting in a whole diagonal, (b) for an $l$ resulting in a wrapped diagonal. In each one of the two examples, the grid points forming $\mathcal{D}_l$ are shown surrounded by red rings.}
  \label{fig:interval_k}
  \vspace{-2mm}
\end{figure}
\begin{equation}
    \label{eq:tilde_s_D}
    \tilde{s}_{\rm d}=(L-1)P+1,\quad,\tilde{s}_{\rm D}=(1+\epsilon)\log(LP)
\end{equation}
Indeed, the first level (of size $(L-1)P+1$) of $\Tilde{\boldsymbol{\alpha}}$ is sensed without compression with a number of measurements equal to $(L-1)P+1$ while $\tilde{s}_{\rm D}$ can be determined thanks to Definition \ref{def:dD_sparsity} and Assumptions \ref{assum:technical} and \ref{assum:technical_diag} and applying the same approach of the proof of Lemma \ref{lem:DS_HS} to $\Tilde{S}_{{\rm D},l}\triangleq\sum_{(\tilde{l},q)\in\mathcal{D}_l}I_{\tilde{l},q}$.

Now, we can write the signal model of sensing $\Tilde{\boldsymbol{\alpha}}$ as
\begin{equation}
    \label{eq:tilde_yp}
    \Tilde{\mathbf{y}}_{\rm p}=\Tilde{\mathbf{M}}_{\rm p}\Tilde{\boldsymbol{\alpha}}
    +\Tilde{\mathbf{w}}_{\rm p},
\end{equation}
where $\Tilde{\mathbf{y}}_{\rm p}=\left[\Tilde{\mathbf{y}}_{{\rm p},0}^{\rm T}\quad\cdots\quad\Tilde{\mathbf{y}}_{{\rm p},(L-1)P}^{\rm T}\right]^{\rm T}$. For each $l$, $\Tilde{\mathbf{y}}_{{\rm p},l}$ is a $N_{\rm p}\times 1$ vector composed of the pilot samples received at the $l$-th DAFT domain position in each of the $N_{\rm p}$ pilot instances. Note that by this definition $\Tilde{\mathbf{y}}_{\rm p}$ is obtained by permuting $\mathbf{y}_{\rm p}$ in \eqref{eq:yp} in accordance with the permutation that gives $\Tilde{\boldsymbol{\alpha}}$ from $\boldsymbol{\alpha}$.  Next, we prove that $\Tilde{\mathbf{M}}_{\rm p}$ has the following Kronecker structure
\begin{equation}
    \label{eq:kron}
    \Tilde{\mathbf{M}}_{\rm p}=\mathbf{I}_{(L-1)P+1}\otimes\Tilde{\mathbf{M}}_{\cal D},
\end{equation}
with $\Tilde{\mathbf{M}}_{\mathcal{D}} = \mathrm{diag}({\rm p}_1\cdots{\rm p}_{N_{\rm p}})\mathbf{F}_{2\lceil\frac{Q}{P}\rceil+1,{\rm p}}\boldsymbol{\Psi}$, $\mathbf{F}_{2\lceil\frac{Q}{P}\rceil+1,{\rm p}}$ is a $N_{\rm p}\times\left(2\lceil\frac{Q}{P}\rceil+1\right)$ partial Fourier measurement matrix and $\boldsymbol{\Psi}$ is a diagonal matrix with unit-modulus entries. We can thus use \cite[Theorem 4.5]{partial_fourier} pertaining to subsampled Fourier matrices to get that for sufficiently large $L$, $Q$, sufficiently small $\delta$, and
\begin{equation}
    \label{eq:cond_N_p}
    N_{\rm p}>O\left(\frac{1}{\delta^{2}}\log^2\frac{1}{\delta}\log\frac{\log (LP)}{\delta}\log(LP)\log\frac{Q}{P}\right)
\end{equation}
the RIP constant $\delta_{\tilde{s}_{\rm D}}$ of $\Tilde{\mathbf{M}}_{\mathcal{D}}$ satisfies $\delta_{\tilde{s}_{\rm D}}\leq\delta$ with probability $1-e^{-\Omega \left(\log{\frac{Q}{P}}\log{\frac{1}{\delta}}\right)}$. The RIP of $\mathbf{I}_{(L-1)P+1}$ trivially satisfies $\delta_{\tilde{s}_{\rm d}}=0$. As for the HiRIP of $\Tilde{\mathbf{M}}_{\rm p}$, we can apply \cite[Theorem 4]{hierarchical} to \eqref{eq:kron} thanks to its Kronecker structure to get
\begin{equation}
    \delta_{s_{\rm d},s_{\rm D}}\leq\delta_{\tilde{s}_{\rm d}}+\delta_{\tilde{s}_{\rm D}}+\delta_{\tilde{s}_{\rm d}}\delta_{\tilde{s}_{\rm D}}\leq\delta
\end{equation}
if $N_{\rm}$ and $\delta$ are as in \eqref{eq:cond_N_p}. This completes the proof.

\bibliographystyle{IEEEtran}
\bibliography{IEEEabrv,Citations}

\end{document}